\documentclass{cccg21}
\usepackage{graphicx,amssymb,amsmath}

\usepackage[colorlinks=true, citecolor=cyan, linkcolor=magenta]{hyperref}
\usepackage[colorinlistoftodos]{todonotes}
\usepackage{tkz-fct}
\usepackage{tikz}
\usepackage{color, colortbl}
\usepackage{mathtools}
\usepackage{thmtools}
\usepackage{thm-restate}
\usepackage{cleveref}

\def\defn#1{\textbf{\textit{\boldmath #1}}}

\newcommand*\samethanks[1][\value{footnote}]{\footnotemark[#1]}

{\makeatletter
 \gdef\xxxmark{%
   \expandafter\ifx\csname @mpargs\endcsname\relax 
     \expandafter\ifx\csname @captype\endcsname\relax 
       \marginpar{xxx}
     \else
       xxx 
     \fi
   \else
     xxx 
   \fi}
 \gdef\xxx{\@ifnextchar[\xxx@lab\xxx@nolab}
 \long\gdef\xxx@lab[#1]#2{\textbf{[\xxxmark #2 ---{\sc #1}]}}
 \long\gdef\xxx@nolab#1{\textbf{[\xxxmark #1]}}
}

\title{Edge-Unfolding Prismatoids: Tall or Rectangular Base}

\author{Vincent Bian\thanks{MIT Department of Mathematics, Cambridge, MA, USA {\tt \{vinvinb,rachanam\}@mit.edu}}
        \and
        Erik D. Demaine\thanks{MIT Computer Science and Artificial Intelligence Laboratory,  Cambridge, MA, USA, {\tt  edemaine@mit.edu}}
        \and
        Rachana Madhukara\samethanks[1]}

\begin{document}
\thispagestyle{empty}
\maketitle

\begin{abstract}
We show how to edge-unfold a new class of convex polyhedra, specifically a new class of prismatoids (the convex hull of two parallel convex polygons, called the top and base), by constructing a nonoverlapping ``petal unfolding'' in two new cases: (1)~when the top and base are sufficiently far from each other; and (2)~when the base is a rectangle and all other faces are nonobtuse triangles. The latter result extends a previous result by O'Rourke that the petal unfolding of a prismatoid avoids overlap when the base is a triangle (possibly obtuse) and all other faces are nonobtuse triangles. We also illustrate the difficulty of extending this result to a general quadrilateral base by giving a counterexample to our technique.
\end{abstract}

\section{Introduction}
A famous open problem known as D\"{u}rer's problem
\cite[Open Problem 21.11, p.~298]{demaine2007geometric}
asks whether every convex polyhedron has an \defn{edge unfolding}, that is,
a set of edges to cut such that the remaining surface unfolds into the plane
without overlap.
Despite the simple statement of the problem, a solution remains elusive.
One approach to making partial progress on this problem is to prove that
special classes of convex polyhedra have edge unfoldings.

One of the simplest yet still-open cases is \defn{prismatoids}, defined as
the convex hull of two parallel convex polygons, called the \defn{top} and
\defn{base} (bottom).
Aloupis \cite{Aloupis-thesis} showed that, if we omit the top and base,
the resulting ``band'' of side faces has an edge unfolding.
The challenge is thus to place the top and base without overlap;
indeed, O'Rourke \cite{BandCounterexample} showed that it is impossible to
simply attach these polygons to an unfolded band without overlap
(a ``band unfolding'').

A simpler goal is to unfold a prismatoid with the top removed,
resulting in a polyhedron homeomorphic to a disk called a
\defn{topless prismatoid}.
At CCCG 2013, O'Rourke \cite{ToplessPrismatoids} constructed an edge unfolding
for any topless prismatoid whose faces other than the base are triangles.
Specifically, the edge unfolding has a strong property called
\defn{petal unfolding}, meaning that it does not cut any of edges incident
the base.

A topless prismatoid can be viewed as the local neighborhood of a single face
on an arbitrary convex polyhedron;
indeed, this work extends past petal unfoldings of a single face and
its edge-adjacent faces (``edge-neighborhood patch'')
on a convex polyhedron \cite{pinciuEdge}
and of ``domes'' where all faces except a base share a single vertex
\cite[Section 22.5.2, p.~319]{demaine2007geometric}
(which introduced petal unfoldings as ``volcano unfoldings'').
On the negative side, O'Rourke \cite{ToplessPrismatoids} showed that the
larger neighborhood of faces sharing a vertex with a single face on a
convex polyhedron (``vertex-neighborhood patch'') does not always have
a nonoverlapping petal unfolding.
On the positive side, O'Rourke \cite{ToplessPrismatoids} showed that such
a neighborhood has a nonoverlapping petal unfolding if the base is a triangle
(possibly obtuse) and all other incident faces are nonobtuse triangles.

The latter result also leads to an edge unfolding of prismatoids with both
the top and base, provided the base $B$ is a triangle (possibly obtuse) and
all other faces (including the top $A$) are nonobtuse triangles.
In this setting, the definition of \defn{petal unfolding} extends to
mean that it does not cut any of edges incident to the base~$B$, and cuts all
but one of the edges incident to the top~$A$.
(Thus, in all cases, the side faces unfold by simple rotation around one
edge of the base~$B$.)
O'Rourke \cite{ToplessPrismatoids} in fact showed that \emph{all}
petal unfoldings of such prismatoids avoid overlap.

\subsection{Our Results}

We expand O'Rourke's methods to encompass a broader family of prismatoids,
showing that petal unfoldings never overlap in two new situations.
Our first result is a step toward O'Rourke's conjecture that the base can be
any convex polygon, provided the other faces are nonobtuse triangles:

\begin{restatable}{thm}{rectangularprismatoid}
  \label{thm1}
  For any prismatoid where the base $B$ is a rectangle and
  all other faces are nonobtuse triangles,
  every petal unfolding avoids overlaps.
\end{restatable}

Our second result takes a different approach, showing that
``tall'' prismatoids always petal unfold, and thus
thin prismatoids form the remaining hard case:

\begin{restatable}{thm}{tallprismatoids}
  \label{thm2}
  For any prismatoid whose top $A$ and base $B$ are sufficiently far apart,
  every petal unfolding avoids overlaps.
  More precisely, petal unfolding avoids overlap if
  $$z \geq \frac{3\pi P_A + 4d_{AB}}{2 \Delta_B},$$
  where
  \begin{itemize}
  \item $z$ is the distance between the planes containing
    the two bases $A,B$ of the prismatoid;
  \item $P_A$ is the perimeter of the top $A$;
  \item $\Delta_B = \pi - \max_i \angle_B(b_i)$ is the smallest turn angle
    in the base~$B$ (in radians);
    
  \item $A'$ is the projection of $A$ onto the plane of~$B$; and
  \item $d_{AB}$ is the diameter of the region $A' \cup B$.
  \end{itemize}
\end{restatable}

We prove these theorems in Sections~\ref{sec:thm1} and~\ref{sec:thm2}
respectively, after covering the relevant background from
\cite{ToplessPrismatoids} in Section~\ref{sec:bg}.
In Section~\ref{sec:quads}, we give counterexamples for
extending our technique of Theorem~\ref{thm1} to general quadrilaterals.

\section{Background}
\label{sec:bg}

We follow the notation given in O'Rourke's paper \cite{ToplessPrismatoids}.
Let $A$ and $B$ be the top and base of the prismatoid, respectively.
Let $a_1, a_2, \ldots, a_m$ and $b_1, b_2, \ldots, b_n$
be the vertices of $A$ and $B$ respectively.
Let $B_i$ be the triangle with one vertex on $A$ and two vertices at $b_i$ and $b_{i + 1}$, where indices are treated modulo~$n$.
Call these triangles \defn{$B$-triangles},
and define \defn{$A$-triangles} similarly.

Consider two consecutive $B$-triangles $B_{i-1} = b_{i - 1} b_i a_j$ and $B_i = b_i b_{i + 1} a_k$ in the unfolding, as in \figurename~\ref{fig:diamond_def}.
Define a diamond region $D_i$ bounded by line segments $b_i a_j$ and $b_i a_k$, and by the rays through $a_j$ and $a_k$ perpendicular to $b_ia_j$ and $b_ia_k$ respectively.
Because all the $A$-triangles are nonobtuse, all the $A$-triangles attached to edges $b_ia_j$ or $b_ia_k$ stay within the region $D_i$.

\begin{figure}
    \centering
    \includegraphics[width=0.45\textwidth]{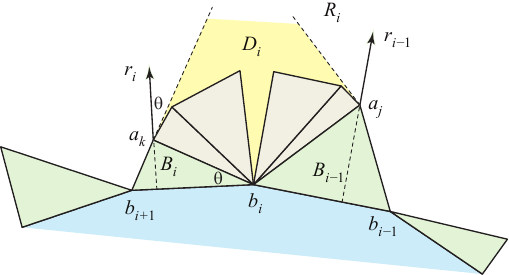}
    \caption{The diamond region $D_i$ and the $A$-triangles it contains.
      [Based on Figure 12(a) of \cite{ToplessPrismatoids}, used with permission.]}
    \label{fig:diamond_def}
\end{figure}

Define a larger wedge region $V_i$ bounded by rays $\overrightarrow{b_ia_j}$ and $\overrightarrow{b_ia_k}$ (and disjoint from $B$, $B_{i-1}$, and $B_i$),
as shown in \figurename~\ref{fig:v_def}.
Wedge $V_i$ contains all the $A$-triangles attached to $b_ia_j$ or $b_ia_k$, as well as the top $A$, should it be attached to one of these $A$-triangles.

\begin{figure}
    \centering
    \includegraphics[width=0.45\textwidth]{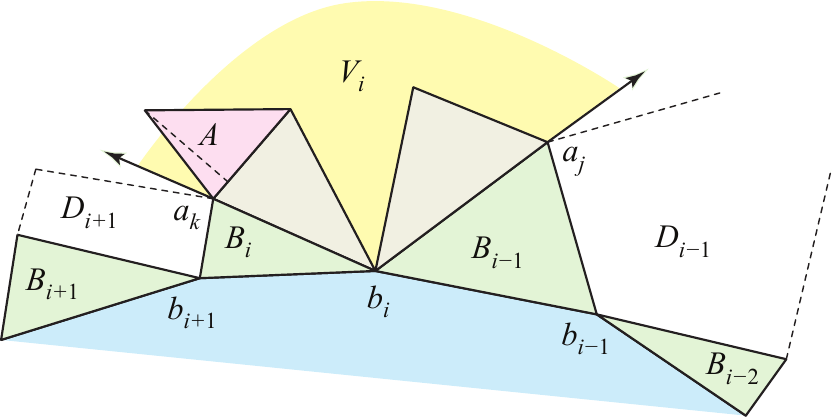}
    \caption{The region $V_i$ containing $A$-triangles and the top $A$.
      [Based on Figure 13 of \cite{ToplessPrismatoids}, used with permission.]}
    \label{fig:v_def}
\end{figure}

%
%
%

\section{Unfolding Rectangular-Base Prismatoids \\ (Proof of Theorem \ref{thm1})}
\label{sec:thm1}

O'Rourke \cite{ToplessPrismatoids} showed that petal unfoldings never overlap for prismatoids with a convex base $B$ and all other faces nonobtuse triangles \emph{provided} that the region $V_i$ does not intersect any $B$-triangles or any diamonds $D_j$ for $j \neq i$ (which contain all other $A$-triangles).
He showed that this property holds when the base $B$ is a triangle (possibly obtuse).
We extend this result to include the case where $B$ is a rectangle,
as in Figure~\ref{fig:rectangle_3d}.

\begin{figure}[!h]
    \centering
    \includegraphics[width=0.4\textwidth]{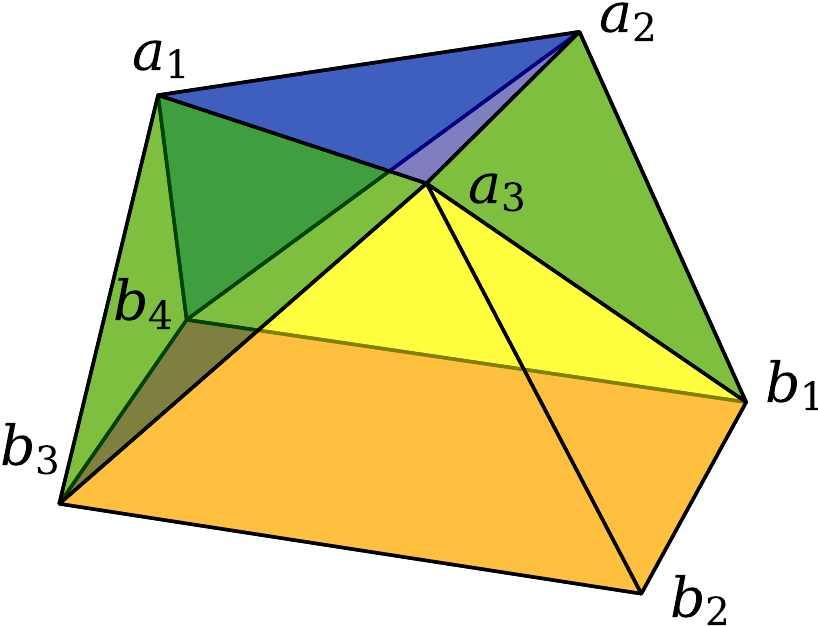}
    \caption{An acutely triangular prismatoid with a rectangular base.}
    \label{fig:rectangle_3d}
\end{figure}


\rectangularprismatoid*

\begin{proof}
    O'Rourke \cite{ToplessPrismatoids} showed that it suffices to prove that $V_i$ does not intersect any $B$-triangles or any diamonds $D_j$ for $j \neq i$.
    He already showed that $V_i$ does not intersect $B_j$, for $i - 2 \leq j \leq i + 1$.
    For a rectangle, this covers all four $B$-triangles.
    Because $B_i$ and $B_{i + 1}$ are acute, the rays bounding $D_{i + 1}$ and $D_{i - 1}$ lie strictly outside $V_i$, so $V_i$ cannot intersect those diamonds.
    Thus, all that remains is to show that $V_i$ does not intersect $D_{i + 2}$.
    
    By symmetry, it suffices to show that $V_1$ does not intersect $D_3$,
    as shown in \figurename~\ref{fig:rectangle_nonintersection}. In fact, we claim that $D_3$ is contained within the region $S$ bounded by rays $\overrightarrow{b_1b_2}$ and $\overrightarrow{b_1b_4}$ containing $b_3$.
    We will show that the line segments and rays bounding $D_3$ never leave~$S$.
    
    \begin{figure}[h!]
        \centering
        \includegraphics[width=0.4\textwidth]{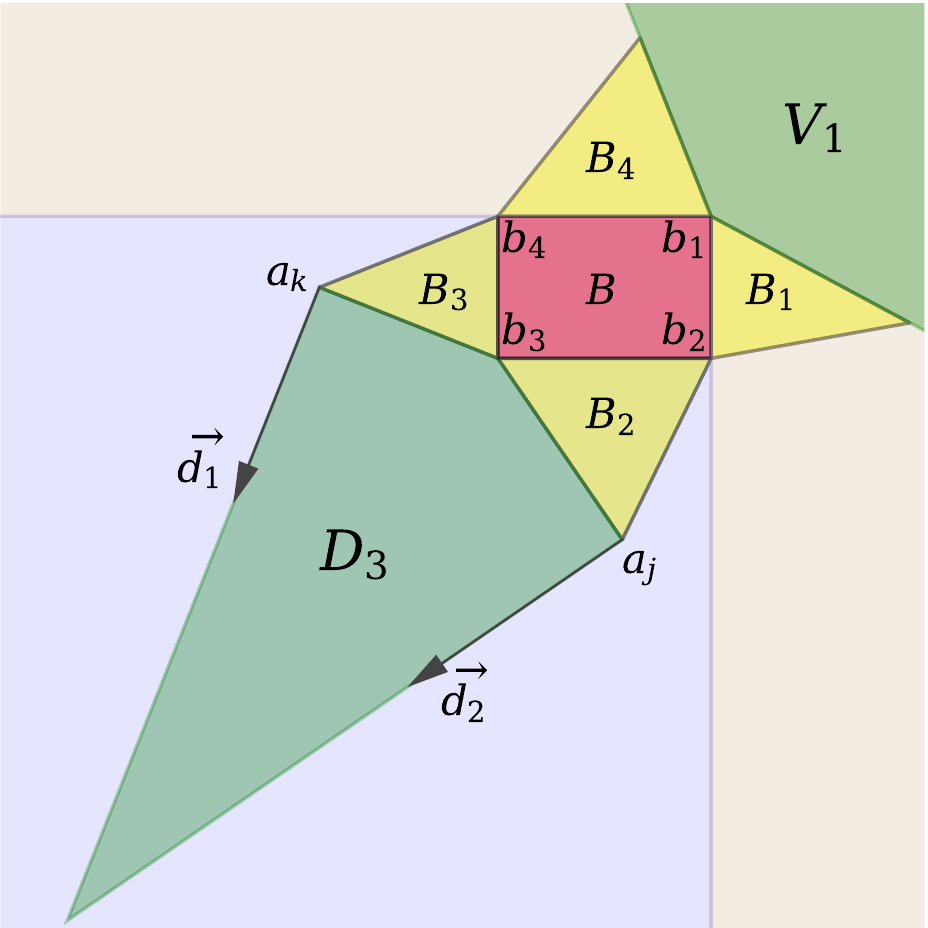}


        \caption{The regions $V_1$ and $D_3$ for the rectangular prismatoid in \figurename~\ref{fig:rectangle_3d}. Note that $D_3$ always lies in the lower left quarter-plane, and $V_1$ always lies in the remaining three quarters of the plane.}
        \label{fig:rectangle_nonintersection}
    \end{figure}
    
    Let $a_j$ and $a_k$ be the apices of triangles $B_2$ and $B_3$, so $D_3$ is bounded by the line segments $b_3a_k$ and $b_3a_j$, and by the rays $\overrightarrow{d_1}$ and $\overrightarrow{d_2}$ perpendicular to $b_3a_k$ and $b_3a_j$ at $a_k$ and $a_j$ respectively.
    
    First, if $b_3a_k$ intersected line $b_1b_4$, then $\angle b_3b_4a_k$ of $B_3$ would be obtuse. Also, $b_3a_k$ cannot intersect ray $b_1b_2$, as it is on the wrong side of line $b_3b_4$. Thus, $b_3a_k$ is contained in $S$. Similarly, $b_3a_j$ is contained in~$S$.
    
    Now consider ray $\overrightarrow{d_1}$. Suppose it intersected $\overrightarrow{b_1b_2}$ at some point $x$. Then, in quadrilateral $b_2b_3a_kx$, we have $\angle b_3a_kx = \angle xb_2b_3 = 90^{\circ}$, meaning $\angle b_2b_3a_k = 180^{\circ} - \angle a_kxb_2 < 180^{\circ}$. However, this would make $\angle b_4b_3a_k = 360^{\circ} - \angle b_2b_3b_4 - \angle b_2b_3a_k > 90^{\circ}$, contradicting $B_3$ being nonobtuse. 
    
    Similarly, suppose that $\overrightarrow{d_1}$ intersects $\overrightarrow{b_1b_4}$ at some point $y$. Then, in triangle $ya_kb_4$, we have $\angle ya_kb_4 < 180^{\circ}$, so $\angle b_4a_kb_3 = 360^{\circ} - \angle ya_kb_4 - \angle b_3a_ky > 360^{\circ} - 180^{\circ} - 90^{\circ} = 90^{\circ}$, contradicting the assumption that $B_3$ is nonobtuse. Hence, $\overrightarrow{d_1}$ never intersects $\overrightarrow{b_1b_2}$ or $\overrightarrow{b_1b_4}$, and is thus contained in~$S$. A similar argument shows that $\overrightarrow{d_2}$ is contained in~$S$.
    
    Finally, we show that $V_1$ intersects $S$ only at point $b_1$. This claim holds because the two rays bounding $V_1$ only ever intersect $\overrightarrow{b_1b_2}$ and $\overrightarrow{b_1b_4}$ at $b_1$. Therefore, all petal unfoldings do not overlap.
\end{proof}

\subsection{Difficulty of Quadrilateral Bases}
\label{sec:quads}

It is natural to hope that Theorem~\ref{thm1} can be extended to all quadrilateral bases, or any convex base. However, our technique above relies on the fact that each angle of $B$ is nonobtuse. Specifically, showing that $V_i$ and $D_{i + 2}$ do not intersect requires the assumption that $\angle b_1b_2b_3 \leq 90^{\circ}$, and $\angle b_1b_4b_3 \leq 90^{\circ}$. Every angle of polygon $B$ is nonobtuse only when $B$ is a rectangle or a nonobtuse triangle, so other quadrilaterals will require a more careful treatment.

Furthermore, the prismatoid $\mathcal{P}_c$, shown in
\figurename~\ref{fig:counterexample_3d} and coordinatized in
Table~\ref{tab:counterexample_coordinates}, is counterexample to the
conjecture that the regions do not overlap when $B$ is a general quadrilateral.
\figurename~\ref{fig:counterexample_regions} shows the overlap.

\begin{figure}[!h]
    \centering
    \includegraphics[width=0.45\textwidth]{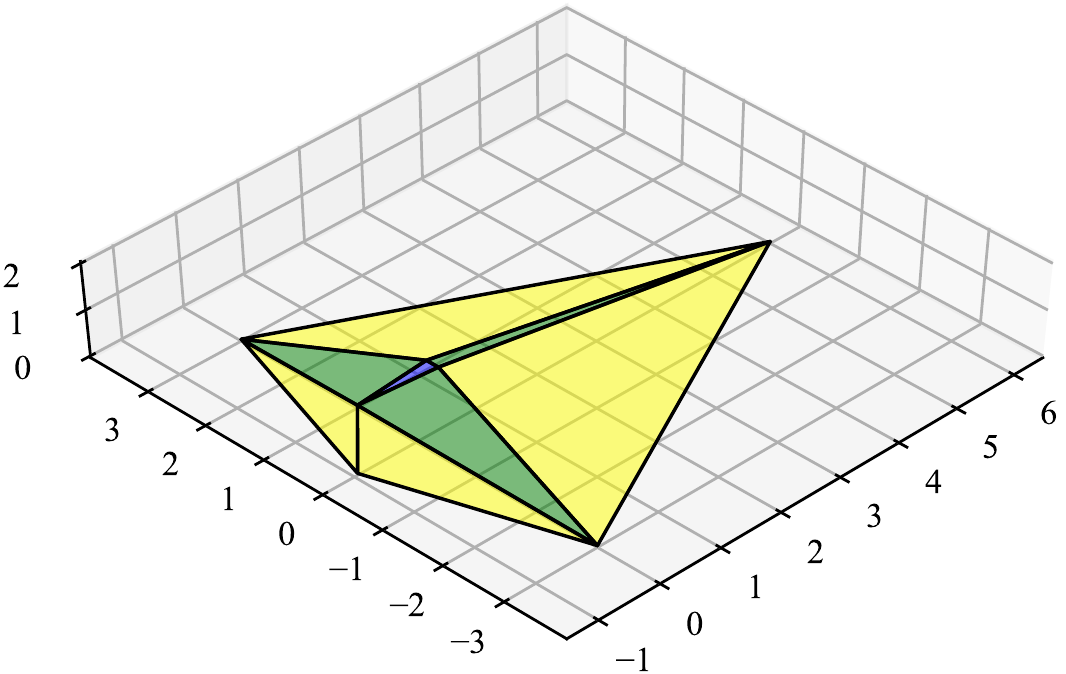}
    \caption{Prismatoid $\mathcal{P}_c$ has a quadrilateral base and all other faces nonobtuse triangles. The largest angle among the triangular faces is $89.7^{\circ}$.}
    \label{fig:counterexample_3d}
\end{figure}

\begin{figure}[!h]
    \centering
    \includegraphics[width=0.4\textwidth]{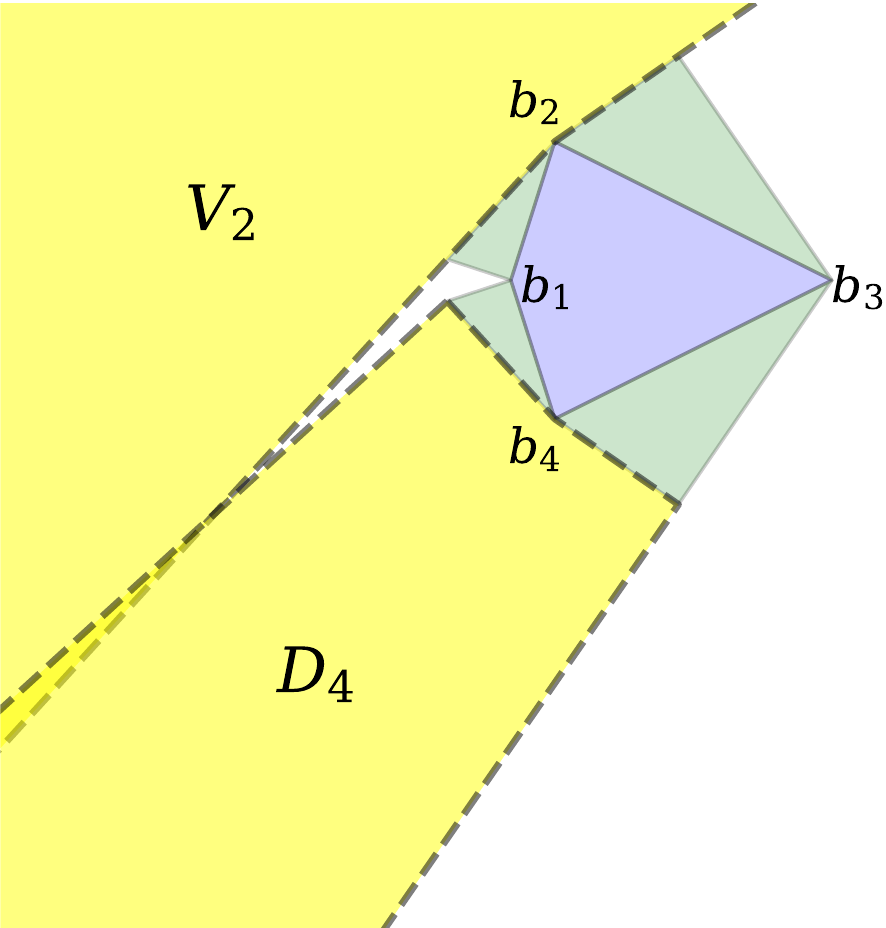}
    \caption{The regions $V_2$ and $D_4$ intersect.}
    \label{fig:counterexample_regions}
\end{figure}

The points of this prismatoid can be moved so that the base $B$ is cyclic (vertices lie on a common circle), forming a new prismatoid $\mathcal{P}_{cyc}$
with coordinates given by Table~\ref{tab:cyclic_coordinates}.
To find the coordinates of $\mathcal{P}_{cyc}$, we used a gradient descent method to minimize $\left|\angle b_1b_2b_3 - 90^{\circ}\right|$ while maintaining that all triangles are nonobtuse.
The overlap of the regions in $\mathcal{P}_{cyc}$ is much more difficult to
see (refer to Figure~\ref{fig:counterexample_regions}):
the angle formed at the intersection point is less than $0.003^{\circ}$.

These examples mean that extending the proof of Theorem~\ref{thm1}, even to just cyclic quadrilaterals, requires a more precise treatment than considering the regions $V_i$ and $D_i$.
On the other hand, all petal unfoldings of $\mathcal{P}_c$ and $\mathcal{P}_{cyc}$ have no overlap, so O'Rourke's conjecture about petal unfoldings with an arbitrary convex base remains plausible.

\begin{table}
    \centering
    \begin{tabular}{c|ccc}
            Point(s) & \multicolumn{3}{c}{Coordinates} \\
            \hline
            $b_1$ & ($-0.95,$& 0.00,& 0.00) \\
            $b_2, b_4$ & (0.00,& $\pm 3.00,$& 0.00) \\
            $b_3$ & (6.00,& 0.00,& 0.00) \\
            $a_1$ & ($-0.90,$& 0.00,& 1.45) \\
            $a_2, a_3$ & (0.30,& $\pm 0.10,$& 1.45)
        \end{tabular}
    \caption{The coordinates of the vertices of $\mathcal{P}_c$.}
    \label{tab:counterexample_coordinates}
\end{table}

\begin{table}
    \centering
    \begin{tabular}{c|ccc}
            Point(s) & \multicolumn{3}{c}{Coordinates} \\
            \hline
            $b_1$ & ($-1.5633,$ & 0.0000, & 0.0000) \\
            $b_2, b_4$ & (0.0000, & $\pm 3.7169,$ & 0.0000) \\
            $b_3$ & (8.8372, & 0.0000, & 0.0000) \\
            $a_1$ & ($-1.5581,$ & 0.0000, & 1.6435) \\
            $a_2, a_3$ & (0.2225, & $\pm 0.0299,$ & 1.6435)
        \end{tabular}
    \caption{The coordinates of the vertices of $\mathcal{P}_{cyc}$.}
    \label{tab:cyclic_coordinates}
\end{table}

\section{Unfolding Tall Prismatoids \\ (Proof of Theorem \ref{thm2})}
\label{sec:thm2}


For a given prismatoid, let $z$ denote the distance between the planes of the top and base. We show that, for prismatoids with large enough $z$, all petal unfoldings avoid overlap.

\tallprismatoids*

\begin{figure}
        \centering
        \includegraphics[width=0.45\textwidth]{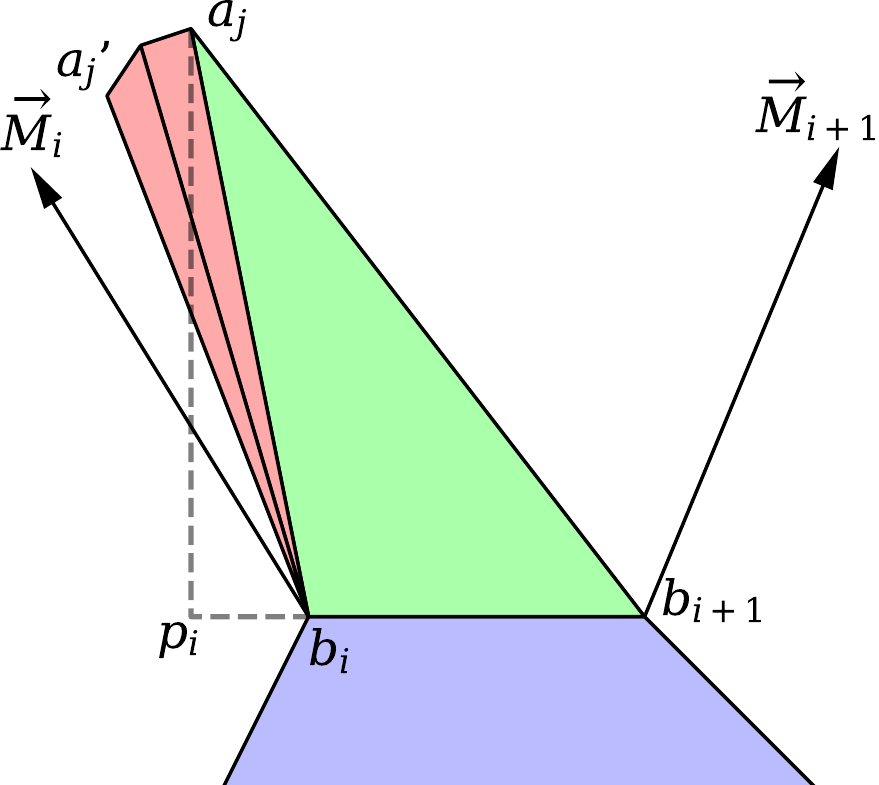}
        \caption{One of the $B$-triangles, along with some $A$-triangles attached to its left. In this configuration, $p_i$ is on the opposite side of $b_i$ as $b_{i + 1}$, so $\angle b_{i + 1}b_ia_j$ is obtuse.}
        \label{fig:tall_diagram}
\end{figure}

\begin{proof}
    We show that, in any petal unfolding, every face that gets attached to a $B$-face $B_i$ will stay in a region $S_i$ bounded by the edge $b_i b_{i + 1}$ and the rays $\overrightarrow{M_i}$ and $\overrightarrow{M_{i + 1}}$ bisecting the exterior angles of $B$ at $b_i$ and $b_{i + 1}$ respectively,
    as shown in \figurename~\ref{fig:tall_diagram}.
    Note that the angle between edge $b_ib_{i + 1}$ and $\overrightarrow{M_i}$ is at least $\frac{\pi}{2} + \frac{\Delta_B}{2}$, and every edge of the form $b_ia_j$ has length at least $z$.

    \medskip

    Let $0 < \ell < 1$ be a constant.
    Consider a $B$-face $B_i$ with vertices $b_ib_{i + 1}a_j.$
    First we claim that the angle $\angle b_{i + 1}b_ia_j$ will be at most $\frac{\pi}{2} + \frac{\Delta_B}{2}\cdot \ell$, as long as $z \geq \frac{2d_{AB}}{\Delta_B \ell}$.

    Consider the projection $p_i$ of $a_j$ onto $b_ib_{i + 1}$. If it lies on the same side of $b_i$ as $b_{i + 1}$, then $\angle b_{i + 1}b_ia_j$ is acute, and we are done. Otherwise, the angle is obtuse, but we can use the fact that the length of $b_ip_i$ is at most $d_{AB}$.

    In this case, we know $\angle b_{i + 1}b_ia_j = \frac{\pi}{2} + \arctan \frac{b_ip_i}{p_ia_j}$. Also $p_ia_j \geq z$, so $$\angle b_{i + 1}b_ia_j \leq \frac{\pi}{2} + \arctan \frac{d_{AB}}{z} \leq \frac{\pi}{2} + \frac{d_{AB}}{z}.$$ Substituting in our assumption that $z \geq \frac{2d_{AB}}{\Delta_B \ell}$, we get that $\angle b_{i + 1}b_ia_j \leq \frac{\pi}{2} + \frac{\Delta_B}{2} \cdot \ell$, as desired.
    
    \medskip

    Second, we show that, as long as $z \geq \frac{3\pi}{2 \Delta_B}P_A \cdot \frac{1}{1 - \ell}$, the angle $\angle a_jb_ia_j'$ subtended by the $A$-triangles attached to edge $a_jb_i$ is at most $\frac{\Delta_B}{3} (1 - \ell)$. We start by bounding the measure of $\angle a_jb_ia_{j + 1}$ for any edge $a_ja_{j + 1}$ of $A$.
    By the Law of Sines, $\frac{\sin \angle a_jb_ia_{j + 1}}{a_j a_{j + 1}} = \frac{\sin \angle a_ja_{j + 1}b_i}{a_jb_i}$, so
    \begin{align*}
        \angle a_jb_ia_{j + 1} &= \arcsin \frac{a_ja_{j + 1} \sin \angle a_ja_{j + 1}b_i}{a_jb_i} \\
        & \leq \arcsin \frac{a_ja_{j + 1}}{z}.
    \end{align*}
    Because $\arcsin x \leq \frac{\pi}{2}x$ for $x \geq 0$, we obtain $\angle a_jb_ia_{j + 1} \leq \frac{\pi}{2} \cdot \frac{a_ja_{j + 1}}{z}$.
    
    The sum of these lengths $a_ja_{j + 1}$ over all $A$-triangles is $P_A$, so the sum of the angles over all $A$-triangles is at most $\frac{\pi P_A}{2z}$. Because the angle $\angle a_jb_ia_j'$ is the sum of $\angle a_jb_ia_{j + 1}$ over some subset of the edges $a_ja_{j + 1}$ of $A$, we can substitute $z \geq \frac{3\pi}{2 \Delta_B}P_A \cdot \frac{1}{1 - \ell}$ to get that $\angle a_jb_ia_j' \leq \frac{\Delta_B}{3} (1 - \ell)$.
    
    \medskip

    Third, we show that, if $z \geq \min\left(\frac{2d}{\Delta_B \ell}, \frac{3\pi}{2 \Delta_B}P_A \cdot \frac{1}{1 - \ell}\right)$, then no matter where the top face $A$ is attached in the unfolding, it will not exit the region $S_i$. We accomplish this by proving that the shortest distance $d_{\min}$ between the point $a_j'$ and the ray $\overrightarrow{M_i}$ is at least $\frac{P_A}{2}$. By the triangle inequality, this means that $A$ cannot intersect $\overrightarrow{M_i}$. Note that this shortest distance is $d_{\min} = b_ia_j' \sin \left(\frac{\pi}{2} + \frac{\Delta_B}{2} - \angle b_{i + 1}b_ia_j - \angle a_jb_ia_j'\right)$.
    
    We know that $b_ia_j' \geq z$, and from our previous results, we know that
    \begin{align*}
        & \, \frac{\pi}{2} + \frac{\Delta_B}{2} - \angle b_{i + 1}b_ia_j - \angle a_jb_ia_j' \\
        \geq &\, \frac{\pi}{2} + \frac{\Delta_B}{2} - \frac{\pi}{2} - \frac{\Delta_B}{2} \cdot \ell - \frac{\Delta_B}{3} (1 - \ell) \\
        = &\, \frac{\Delta_B}{6}(1 - \ell).
    \end{align*}
    Using the fact that $\sin x \geq \frac{2x}{\pi}$ for $0 \leq x \leq \frac{\pi}{2}$, we obtain $$d_{\min} \geq \frac{3\pi}{2\Delta_B}P_A \cdot \frac{1}{1 - \ell} \cdot \frac{2}{\pi} \cdot \frac{\Delta_B}{6}(1 - \ell) = \frac{P_A}{2},$$ as desired.
    
    Repeating this argument for every side $b_ia_j$ of every $B$-triangle, we obtain that, if $$z \geq \min\left(\frac{2d}{\Delta_B \ell}, \frac{3\pi}{2 \Delta_B}P_A \cdot \frac{1}{1 - \ell}\right),$$ then no petal unfolding of $\mathcal{P}$ can overlap. This lower bound is minimized when the two inputs to the $\min$ are equal. This occurs when $\ell = \frac{4d}{4\pi P_A + 4d}$, which when substituted yields the desired $z \geq \frac{3\pi P_A + 4d_{AB}}{2 \Delta B}$.
\end{proof}

The most room for improvement in this proof is the second step's bound $z \geq \frac{3\pi}{2\Delta_B}P_A \cdot \frac{1}{1 - \ell}$, as it is impossible for all the $A$-triangles to be attached to a single point on $A$.

\section*{Acknowledgments}   
This work began as a final project in an MIT class on Geometric Folding Algorithms (6.849, Fall 2020). We thank Professor Joseph O'Rourke for his guidance on this project. We also thank Yevhenii Diomidov for helpful suggestions.


\small




\bibliographystyle{plain}
\bibliography{references.bib}


\end{document}